\documentclass[11pt,a4paper]{article}
\usepackage{amsfonts,amssymb,amscd}
\usepackage{latexsym, graphicx}
\usepackage[affil-it]{authblk}
\usepackage{algorithmic}
\usepackage{algorithm}

\usepackage[margin=3cm]{geometry}

\usepackage{enumerate}

\def\qed{\hfill $\Box$\vspace{2ex}}

\newtheorem{theorem} {Theorem}[section]
\newtheorem{lemma}[theorem]{Lemma}
\newtheorem{observation}[theorem]{Observation}
\newtheorem{corollary}[theorem]{Corollary}
\newtheorem{problem}{Problem}[section]

\usepackage{xcolor}
\newenvironment{proof}{\noindent {\it Proof:~}}{\hfill $\Box$\smallskip\par}

\def\inst#1{$^{#1}$}
\begin{document}

\title{On the structure of (pan, even hole)-free graphs\thanks{Research
support by Natural Sciences and Engineering Research Council of
Canada.}}
\author{Kathie Cameron\inst{1}
 \and Steven Chaplick\inst{2}\thanks{Research partially supported by the European Science Foundation project EUROGIGA GraDR.}
  \and Ch\'inh T. Ho\`ang\inst{3}
}
\date{}
\maketitle
\begin{center}
{\footnotesize

\inst{1} Department of Mathematics, Wilfrid Laurier University,
Waterloo, Ontario,  Canada, N2L 3C5\\
\texttt{kcameron@wlu.ca}

\inst{2} Lehrstuhl f\"ur Informatik I
Universit\"at W\"urzburg, Am Hubland
D-97074 W\"urzburg, Germany\\
\texttt{steven.chaplick@uni-wuerzburg.de}

\inst{3} Department of Physics and Computer Science, Wilfrid
Laurier University,
Waterloo, Ontario,  Canada, N2L 3C5\\
\texttt{choang@wlu.ca} }

\end{center}

\begin{abstract}
A hole is a chordless  cycle with at least four vertices. A pan is
a graph which consists of a hole and a single vertex with
precisely one neighbor on the hole. An even hole is a hole with
an even number of vertices. We prove that a (pan, even hole)-free
graph can be decomposed by  clique cutsets into essentially unit
circular-arc graphs. This structure theorem is the basis of our
$O(nm)$-time certifying algorithm for recognizing (pan, even
hole)-free graphs and for our $O(n^{2.5}+nm)$-time algorithm to optimally
color them. Using this structure theorem, we show that the tree-width
of a (pan, even hole)-free graph is at most 1.5 times the clique
number minus 1, and thus the chromatic number is at most 1.5 times the
clique number.
\end{abstract}

\section{Introduction}\label{sec:introduction}
A {\em hole} is a chordless cycle with at least four vertices. A
graph is  {\em chordal} if it does not contain a hole as an
induced subgraph. Chordal graphs are well-studied and have a
number of interesting structural properties (see
\cite{BerChv1984,Gol1980}). For example, it is known
\cite{Dir1961} that every chordal graph contains a {\it
simplicial} vertex; i.e., a vertex whose neighborhood induces a
clique. Based on this, a largest clique, a minimum coloring, a
largest stable set, and a minimum partition into cliques of a
chordal graph can be found in polynomial time \cite{Gavril1972}.

An {\em even hole} is a hole with an even number of vertices. A
graph is {\em even-hole-free} if it does not contain an even hole
as an induced subgraph. Even-hole-free graphs generalize chordal
graphs and analogous properties have been found.
A largest clique of an even-hole-free graph can be found in
polynomial time \cite{actv, AddChu2008, daSV}. However, it is not known whether even-hole-free
graphs can be optimally colored in polynomial time.

The {\it claw} is the graph with vertices $a,b,c,d$ and edges
$ab,ac,ad$. As usual, $n$ (respectively, $m$) denotes the number of vertices
(respectively, edges) of the input graph $G$. We give an $O(n^{2.5}+nm)$-time
algorithm to color (claw, even hole)-free graphs,
providing a contrast to the well-known result \cite{Hol1981} that it is NP-hard
to optimally color claw-free graphs. Our techniques actually apply to a
larger class of graphs which we will now define. An {\it atom  }
is a connected graph without a clique cutset. A {\it pan} is a graph
which consists of a hole and a single vertex with precisely one
neighbor on the hole. Let ${\cal C}$ denote the class of graphs
$G$ such that each atom of $G$ is (pan, even hole)-free. In this
paper, we will give an $O(n^{2.5}+nm)$-time algorithm to color a graph in
${\cal C}$, and an $O(nm)$-time certifying algorithm for
recognition of graphs in ${\cal C}$ and recognition of (pan, even
hole)-free graphs. Pan-free graphs have been studied previously
regarding: establishing the perfectness of (pan, odd hole)-free graphs \cite{Ola1989},
and providing a polynomial-time algorithm to find a largest weight stable
set first on a subclass of pan-free graphs \cite{Des1993} and then the whole
class of pan-free graphs \cite{BraLoz2010}.
The latter two
papers use the term ``apple" for ``pan".

In Section~\ref{sec:definitions}, we will cover the relevant background
and state our main results. In
Section~\ref{sec:properties}, we will prove that a (pan, even
hole)-free graph can be decomposed by the well-studied clique
cutset decomposition into, essentially, ``unit circular-arc
graphs". This structural result is the foundation of our
polynomial-time algorithms. In Section \ref{sec:coloring}, we give
our $O(n^{2.5}+nm)$-time algorithm for coloring the graphs in ${\cal C}$.
In Section ~\ref{sec:recognition}, we discuss our $O(nm)$-time
algorithm to recognize if a graph is in ${\cal C}$. In
Section~\ref{sec:tw}, we show that the tree-width of a (pan, even
hole)-free graph is at most 1.5 times the clique number minus 1,
and thus the chromatic number is at most 1.5 time the
clique number. In Section~\ref{sec:conclusions}, we discuss open problems
related to our work.
\section{Background and results}\label{sec:definitions}
In this section, we discuss the relevant background and give
the definitions necessary to state our main results.  Let $G$ be a
graph. For a subset $S$ of the vertices of $G$, we use $G[S]$
to denote the subgraph of $G$ induced by $S$.  A {\em clique cutset} of $G$ is a set $S$ of vertices where
$G[S]$ is a clique whose
removal increases the number of components of $G$. The following theorem is well known.
\begin{theorem}\label{thm:dirac}{\rm \cite{Dir1961}}
Every chordal graph is either a clique or contains a clique
cutset.
\end{theorem}
Recall that a vertex is \emph{simplicial} if its neighborhood
induces a clique,  and a vertex is \emph{bi-simplicial} if its
neighborhood can be partitioned into two cliques (i.e., its
neighborhood induces the complement of a bipartite graph). It
follows from Theorem~\ref{thm:dirac} that every chordal graph
contains a simplicial vertex.

Let $d_G(x)$ denote the degree of a vertex $x$ in a graph $G$.
When the context is clear, we will write $d(x)$ to mean $d_G(x)$.
Let $\chi(G)$ (respectively, $\omega(G)$) denote the chromatic
number (respectively, the clique number, i.e., the number of
vertices in a largest clique) of $G$. If $v$ is a simplicial
vertex of $G$, then $\chi(G) = \max (\chi(G-v), d(v) +1)$ and
$\omega(G) = \max (\omega(G-v), d(v)+1)$. An analogous property
was established for even-hole-free graphs. (Recall a hole is {\it
even} ({\it odd}) if it has an even (odd) number of vertices.)
\begin{theorem}\label{thm:A}{\rm \cite{AddChu2008}} Every
even-hole-free graph contains a bi-simplicial vertex.
\end{theorem}
Theorem~\ref{thm:A} implies that for even-hole-free graphs $G$,
a largest clique can be found in polynomial time and that $\chi(G) \leq 2 \omega(G) - 1$. However,
it is not known if the coloring problem can be solved in polynomial time for even-hole-free
graphs.

The clique (respectively, chordless cycle, chordless path) on $t$
vertices is denoted by $K_t$ (respectively, $C_t$, $P_t$). Recall that the
{\it claw} is the graph with vertices $a,b,c,d$ and edges
$ab,ac,ad$;  vertex $a$ is the {\em center} of the claw.
Let  $F$ be a graph and let $\mathcal{F}$ be a family of graphs. We say that a graph
$G$ is {\em $F$-free} if $G$ does not contain
an induced subgraph isomorphic to $F$ and $G$ is {\em $\mathcal{F}$-free}
if $G$ does not contain an induced subgraph isomorphic to any graph in~$\mathcal{F}$.
In particular,  $G$
is (claw, even hole)-free if $G$ does not contain a claw or an
even hole as an induced subgraph.

Theorem~\ref{thm:A} implies that
an even-hole-free graph contains a vertex that is not the center
of a claw. This suggests that even-hole-free graphs such that no vertex
is the center of a claw (i.e., (claw, even hole)-free graphs)
might have interesting structure. Indeed, our results show that
(claw, even hole)-free graphs can be decomposed by the clique
cutset decomposition into (essentially) unit circular-arc graphs.
Our results actually apply to a larger class of graphs that we
will define later in this section (see Theorem \ref{thm:structure}).
\begin{figure}[h]
\centering
\includegraphics[scale=1.5]{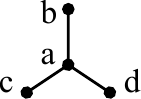}
\caption{The \emph{claw} with center $a$.}
\label{fig:small}
\end{figure}

Consider  the following procedure to decompose a graph  $G$. If
$G$ has a clique cutset $C$, then $G$ can be decomposed into
subgraphs $G_1 = G[V_1]$ and $G_2 = G[V_2]$ where $V = V_1 \cup
V_2$ and $C = V_1 \cap V_2$. Note: there is no edge between $(G_1
- C$) and $(G_2 - C)$.  Given minimum colorings of $G_1$ and $G_2$, we
can obtain a minimum coloring of $G$ by identifying the coloring
of $C$ in $G_1$ with that of $C$ in $G_2$. In particular, we have
$\chi(G)= \max(\chi(G_1),\chi(G_2))$. If $G_i$ ($i \in \{1,2\}$)
has a clique cutset, then we can recursively decompose $G_i$ in
the same way.  This decomposition can be represented by a binary
tree $T(G)$ whose root is $G$ and where the two children of $G$ are
$G_1$ and $G_2$, which are in turn the roots of subtrees
representing the decompositions of $G_1$ and $G_2$. Each leaf of
$T(G)$ corresponds to an induced subgraph of G that contains no
clique cutset; such an induced graph is called an {\em atom} of
$G$. Algorithmic aspects of the clique cutset decomposition are
studied in \cite{Tar1985} and \cite{Whi1984}. In particular, the
decomposition tree $T(G)$ can be constructed in $O(nm)$ time such
that the total number atoms is at most $n-1$ \cite{Tar1985}. We
have seen in the discussion above that the clique cutset decomposition
can be used to color a graph. With $G$, $G_1$, and $G_2$ defined as
above, $G$ contains an even hole (or odd hole) if and only if
$G_1$ or $G_2$ does. Thus, the clique cutset decomposition can
also be used to find an even hole (or odd hole), if one exists.

Recall that a {\it pan} is the graph obtained from taking a $C_k$ with
$k \geq 4$, adding another vertex $x$, and joining $x$ to a vertex
$y$ of the $C_k$ by an edge. The edge $xy$ is called the {\em
handle} of the pan.  We will show that (pan, even hole)-free atoms
have very special structure. This structure allows us to solve the
recognition and coloring problems. To describe this structure, we will need to
introduce more definitions.

A graph $G$ is a {\em circular-arc} graph if there is a bijection
between its vertices and a set $A$ of arcs on a circle such that
two vertices of $G$ are adjacent if and only if the two
corresponding arcs of $A$ intersect. A circular-arc graph is {\it
proper}  if no arc contains another. Additionally, $G$ is a {\em
unit} circular-arc graph if every arc of $A$ has the same length.
It is easy to see that unit circular-arc graphs are proper and
that proper circular-arc graphs are claw-free and hence pan-free.

Let $A$ and $B$ be two disjoint sets of vertices. We say $A$ is {\em
$B$-null} if there is no edge between $A$ and $B$, and $A$ is {\em
$B$-complete} if every possible edge between $A$ and $B$ is
present. For a vertex $x$, $N_G(x)$ denotes the set of neighbors
of $x$ in $G$. When the context is obvious, we use $N(x)$ for $
N_G(x)$.  For a set $X$ of vertices, $N(X)$ denotes the set of
vertices outside $X$ that have neighbors in $X$. A vertex $a$ {\em
dominates} a vertex $b$ if $(N(b) - \{a\}) \subseteq  N(a)$.
Vertex $a$ {\em strictly dominates} vertex $b$ if $(N(b) - \{a\})
\subsetneq N(a)$. Two vertices are {\em comparable} if one dominates
the other. The domination relation is transitive, that is, if $a$
dominates $b$ and $b$ dominates $c$, then $a$ dominates $c$. Thus,
given a set $X$ of vertices such that any two vertices in $X$ are
comparable, there is a total order $\prec$ on $X$ such that $a
\prec b$ whenever $a$ dominates $b$. We call such order a {\it
domination order}.  Two vertices $a$ and $b$ are comparable in $X$ if
they are comparable in the subgraph induced by $X \cup \{a,b\}$.

Let $A$ and $B$ be two vertex-disjoint graphs. The {\em join} of $A$ and
$B$ is the graph $C$ obtained from $A$ and $B$ by adding every
edge between the vertices of $A$ and those of $B$; thus, in $C$
the vertices of $A$ are $B$-complete and vice versa.

We now state our decomposition theorem for the class ${\cal C}$.

\begin{theorem}\label{thm:structure}
If $G$ is a connected graph in ${\cal C}$ (i.e. every atom of $G$
is (pan, even hole)-free), then
\begin{enumerate}[(i)]
  \item $G$ is a clique, or
  \item $G$ contains a clique cutset, or
  \item $G$ is  a unit circular-arc graph, or
  \item $G$ is the join of a unit circular-arc graph and a
clique.
\end{enumerate}
\end{theorem}

In \cite{CamEsc2007}, a polynomial-time algorithm is given for
finding an even hole (or, odd hole) in a circular-arc graph. In
\cite{OrlBon1991} or combining \cite{LS2008} and \cite{ShihHsu},
polynomial-time algorithms are given for
finding an optimal coloring of a unit circular-arc graph. In
Sections~\ref{sec:coloring} and~\ref{sec:recognition}, we discuss
how these results can be used to color and recognize
both the class ${\cal C}$ and the class of (pan, even hole)-free graphs.
In Section~\ref{sec:properties}, we will
establish structural properties of (pan, even hole)-free graphs.
Specifically, we will prove Theorem 2.3.

Our main result is:
\begin{theorem}\label{MainResult}
Given a graph $G$, a pan or even hole of $G$, if one exists,
can be found in $O(nm)$ time.
\end{theorem}

We end this section with a discussion on the relationship between
even-hole-free graphs and $\beta$-perfect graphs, which were
introduced in \cite{MarGas1996} and are defined as follows. Let  $\delta(G)$ denote the
minimum degree of a vertex of a graph $G$. Order the vertices of
$G$ as $v_1, \ldots, v_n$ where $v_i$ has minimum degree in
$G[v_i, \ldots, v_n]$. Greedily color $G$ starting from $v_n$,
i.e., $v_i$ is given the smallest color distinct from its
neighborhood in $G[v_i, \ldots, v_n ]$. The number of colors used
is at most maximum$\{\delta(H) + 1 : H$ is an induced subgraph of
$G\}$, which is denoted $\beta(G)$.  Thus for any graph $G$, we
have $\chi(G) \le \beta(G)$. A graph is defined to be
$\beta$-perfect, if for every induced subgraph $F$ of $G$,
$\chi(F)=\beta(F)$. An even hole $C_{2k}$ has $\chi(C_{2k})=2$ and
$\beta(C_{2k})=3$.  It follows that $\beta$-perfect graphs
are even-hole-free.  A diamond is the complete graph on four
vertices with an edge removed. In \cite{KloMul2009}, it is proved
that (diamond, even hole)-free graphs are $\beta$-perfect.  In
\cite{MarGas1996}, the authors gave the graph of
Figure~\ref{fig:non-beta-perfect} as an example of an
even-hole-free graph which is not $\beta$-perfect. This graph is
claw-free and hence pan-free, so it follows that (claw, even
hole)-free and thus (pan, even hole)-free graphs need not be
$\beta$-perfect.
\begin{figure}[h]
\centering
\includegraphics[scale=1.2]{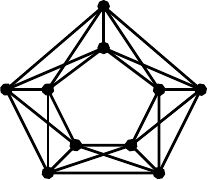}
\caption{A claw-free, non-$\beta$-perfect graph}
\label{fig:non-beta-perfect}
\end{figure}
\section{Properties of (pan, even hole)-free graphs}
\label{sec:properties}
In this section we will prove our structure theorem (Theorem
\ref{thm:structure}). We will actually prove a stronger but
more technical result which implies Theorem
\ref{thm:structure} (see Theorem \ref{thm:structure_buoy}).
We separate the discussion into two
subsections. In the first we consider a special substructure of a
graph which generalizes holes: we call this substructure a
``buoy". In the second we prove that (pan, even hole)-free
graphs decompose into buoys via clique cutsets.

\subsection{Buoys}
\label{sec:buoy}
To motivate our buoys we start with a key observation
regarding the structure around a hole in a pan-free graph.
\begin{observation}\label{obs:neighbors}
Let $G$ be a pan-free graph. Let $C$ be a hole of $G$ of length
at least five and let $x$ be a vertex outside~$C$.
\begin{enumerate}[(i)]
 \item If $x$ has a neighbor $v$ in $C$, then some $u \in C$ is
 adjacent to both $v$ and $x$.
 \item If $x$ has exactly three neighbors $v_1, v_2, v_3$ in
 $C$, then $v_1, v_2, v_3$ forms a path in $C$.
 \item If $x$ has exactly four neighbors in $C$, then $G$
 contains an even hole.
 \item If $x$ has at least five neighbors in $C$, then $x$ is
 $C$-complete.
  \item If $G$ is even-hole-free, then $x$ has 2, 3, or $\ell$ neighbors in $C$, where
  $\ell$ is the length of~$C$.
\end{enumerate}
\end{observation}

\begin{proof}
Enumerate the vertices of $C$ in the cyclic order as $v_0, v_1,
\ldots, v_{\ell -1}$. Suppose that (i) is false for a neighbor $v$
of $x$. We may assume $v= v_0$, i.e., $x$ is adjacent to $v_0$,
and non-adjacent to $v_1$ and $v_{\ell -1}$. Vertex $x$ must have
another neighbor in $C$, for otherwise $x$ and $C$ form a pan. Let
$k$ be the smallest subscript, different from $0$, such that $v_k$
is a neighbor of $x$. If $k = \ell -2$, then $\{v_1, v_0, v_{\ell
-1}, v_{\ell -2}, x \}$ induces a pan. If $k \not= \ell -2$, then
$\{v_{\ell -1}, v_0, v_1, \dots,   v_k, x \}$ induces an pan. We
have established (i).
Now (ii) follows immediately from (i). For (iii) suppose $C$ is
an odd hole (otherwise, we are done). Now, by (i), these four
vertices either form single sub-path $v_i v_{i+1} v_{i+2} v_{i+3}$
of $C$ or two non-adjacent paths $v_i v_{i+1}$ and $v_t v_{t+1}$
such that $C$ can be written as $v_{i+1} P_1 v_t v_{t+1} P_2 v_i$
(for non-empty paths $P_1,P_2$). In the former case, an even hole
is induced by $\{x\} \cup C - \{v_{i+1} v_{i+2}\}$. In the latter
case, one of $x v_{i+1} P_1 v_{t}$ and $x v_{t+1} P_2 v_i$ is an
even hole. We now prove (iv). Suppose $x$ has at least five
neighbors in $C$ but is not $C$-complete (this implies $\ell >
5$). We may assume $x$ is adjacent to $v_0$ and non-adjacent to
$v_1$. Let $k$ be the smallest subscript, different from $0$, such
that $v_k$ is a neighbor of $x$. By (i), $x$ is adjacent to
$v_{k+1}$ and $v_{\ell -1}$. Since $x$ has at least five neighbors
in $C$, $x$ is adjacent to a vertex $v_t$ with $k+1 < t < \ell
-1$. But now $\{v_t, x, v_0, v_1, \ldots, v_k\}$ induces a pan.
Part (v) follows from (i)--(iv).
\end{proof}
For the purpose of finding a forbidden induced subgraph for
recognition of  class ${\cal C}$, we will now reformulate the
results of Observation~\ref{obs:neighbors} into their algorithmic
counter-parts. From the proof  of Observation~\ref{obs:neighbors},
we can extract a linear-time algorithm to find a pan or even hole
of an input graph when one of the conditions (i)-(v) fails.
\begin{observation}\label{obs:find-neighbors}
Let $G$ be a graph, let $C$ be a hole of $G$ of length at least five,
and let $x$ be a vertex outside $C$. If $x$ fails to satisfy (i)--(v)
of Observation~\ref{obs:neighbors}, then $G$ contains a pan or an
even hole, and such an induced graph can be found in linear time.
\hfill $\Box$
\end{observation}

We generalize the presence of a length $\ell \geq 5$ hole (and
Observation \ref{obs:neighbors}) in a graph to the presence of an
\emph{$\ell$-buoy} (defined as follows). For $\ell \geq 5$, an
{\em $\ell$-buoy} $B$ is a collection of sets $B_0, B_1, \ldots,
B_{\ell-1}$ of vertices of $G$ such that each $B_i$ induces a
clique, each vertex in $B_i$ has a neighbor in $B_{i+1}$ and one
in $B_{i-1}$, and there are no edges between $B_i$ and $B -
(B_{i-1} \cup B_i \cup B_{i+1})$, with subscripts taken modulo
$\ell$ (see Figure~\ref{fig:5-buoy} for an example); the sets
$B_i$ are called the {\em bags} of the buoy; a buoy is \emph{odd} or
\emph{even} depending on whether the number of bags ($\ell$) is odd or even.
We also refer to $G[B]$ as a buoy. A {\em skeleton} of
$B$ is a hole containing one vertex of each $B_j$, $j \in \{0,
\ldots, \ell-1\}$. A buoy $B$ in a graph $G$ is said to be
\emph{full} when it includes every vertex of $G$. Due to the
cyclic structure of $\ell$-buoys, when we refer to a bag $B_i$ of
an $\ell$-buoy, we always mean the bag $B_{i~({\rm mod}~\ell)}$.
We will see that when $G$ is $C_4$-free its buoys are circular-arc
graphs (see Theorem \ref{thm:buoy-square-circarc}). Similarly,
when $G$ is (pan, even hole)-free its buoys are unit
circular-arc  graphs (see Theorem \ref{thm:buoy-unit-circ}).
\begin{figure}[h]
\centering
\includegraphics[scale=1.2]{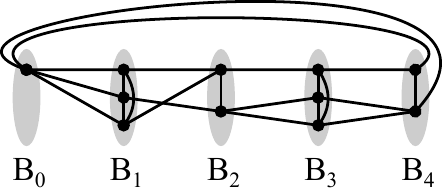}
\caption{An example of a 5-buoy.}
\label{fig:5-buoy}
\end{figure}
\begin{observation}\label{obs:hole}
Let $G$ be an even-hole-free graph having an odd $\ell$-buoy with bags $B_0,
\ldots, B_{\ell-1}$. Consider a path $b_0,
b_1, \ldots , b_{k-1}$ with $1 \leq k \leq \ell$ where $b_i \in
B_i$ for $i = 0, \dots, k-1$. Then where
$z= \min\{k-1,\ell-2\}$, $\{b_0, \ldots, b_{z}\}$ belongs to a skeleton of $B$.
\end{observation}
\begin{proof}
By definition of the buoy, there is an induced path $b_{k-1}, b_k,
\ldots, b_{\ell-1}$ such that $b_j \in B_j$ for $j=k, k+1, \ldots,
\ell-1$. We may assume $b_0$ is not adjacent to $b_{\ell-1}$, for
otherwise we are done. Let $y$ be a neighbor of $b_0$ in
$B_{\ell-1}$. We have $b_{\ell-2} y \in E(G)$, for otherwise
$\{b_0, b_1 \ldots, b_{\ell-2}, b_{\ell-1}, y\}$ induced an even
hole. If $k = \ell$, then $\{b_0, \ldots, b_{k-2}, y\}$ induces a
skeleton; if $k < \ell$, then $\{b_0, \ldots, b_{k-2}, b_{k-1},
\ldots, b_{\ell-2}, y\}$ induces a skeleton.
\end{proof}
\begin{figure}[h]
\centering
\includegraphics[scale=1.2]{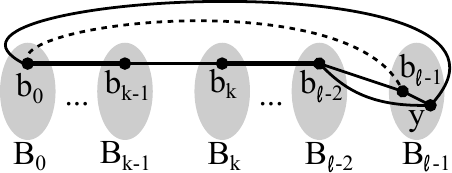}
\caption{The skeleton containing the path $b_0, \ldots, b_{z}$
where $z = \min\{k-1,\ell-2\}$ as in the proof of Observation
\ref{obs:hole}. Note: the bold edges connecting $b_0$ to
$b_{k-1}$ and $b_k$ to $b_{\ell-2}$ correspond to the paths
connecting these vertices.}
\end{figure}

\begin{corollary}\label{cor:hole}
Let $G$ be an even-hole-free graph having an odd $\ell$-buoy with bags $B_0,$
$\ldots,$ $B_{\ell-1}$. Let $P$ be a path with
vertices $p_i, p_{i+1}, \ldots, p_{i+k}$ where $0 \leq k \leq
\ell-2$ and $p_j \in B_j$ for all $j$ (with the subscripts taken
modulo $\ell$), then $P$ belongs to a skeleton of $B$.
\hfill$\Box$
\end{corollary}
From the proof of Observation~\ref{obs:hole}, we can extract a
linear-time algorithm to establish the following observation.
\begin{observation}\label{obs:find-hole}
Let $G$ be a graph having an odd $\ell$-buoy with bags $B_0, \ldots ,
B_{\ell-1}$. Consider a path $b_0, b_1, \ldots
, b_{k-1}$ with $1 \leq k \leq \ell$ where $b_i \in B_i$ for $i =
0, \dots, k-1$. Then there is a linear-time algorithm that either
finds an even hole, or a skeleton containing the set $\{b_0,
\ldots,$ $ b_{z}\}$ (for a given $z = \min\{k-1,\ell-2\}$). \hfill
$\Box$
\end{observation}

\begin{observation}\label{obs:buoy-square}
Let $G$ be a $C_4$-free graph. Let $B$ be an $\ell$-buoy of $G$
with bags $B_0, \ldots, B_{\ell-1}$. Then any two vertices $a$ and $b$
in the same $B_i$ are comparable in $B_{i-1}$. By symmetry, $a$
and $b$ are comparable in $B_{i+1}$.
\end{observation}
\begin{proof}
Let $a$ and $b$ be two vertices in $B_i$. Suppose they are not
comparable in $B_{i-1}$. Then there are vertices $x,y \in B_{i-1}$
with $xa, yb \in E(G)$ and $xb, ya \not\in E(G)$. Now, the four
vertices $a,b,x,y$ form a $C_4$.
\end{proof}
From the proof of Observation~\ref{obs:buoy-square}, we can
extract a linear-time algorithm to establish the following
observation.
\begin{observation}\label{obs:find-buoy-square}
Let $G$ be a  graph. Let $B$ be an $\ell$-buoy of $G$ with bags
$B_0, \ldots,$ $B_{\ell-1}$. If two vertices $a$ and $b$ in the same $B_i$
are not comparable in $B_{i-1}$, then $G$ contains a $C_4$, and
this $C_4$ can be found in linear time. \hfill $\Box$
\end{observation}

\begin{theorem}\label{thm:buoy-square-circarc}
If $B$ is an $\ell$-buoy of a $C_4$-free graph $G$, then $B$ is a
circular-arc graph.
\end{theorem}
\begin{proof}
Let $B_0, \ldots, B_{\ell-1}$ be the bags of $B$. We construct a
circular-arc representation of $G$ as follows. First we partition
the circle into $\ell$ arcs of equal length and label the boundary
points of these arcs as $(0), (1), \ldots, (\ell-1)$ in clockwise
order. By Observation~\ref{obs:buoy-square} the vertices of $B_i$
can be partitioned and ordered by neighborhood inclusion with
respect to $B_{i+1}$.

That is, we let $B_i = \bigcup_{j=1}^{t_i} B_{i,j}$ such that for
$v,v' \in B_{i,j}$, $N(v) \cap B_{i+1} = N(v') \cap B_{i+1}$ and
for $u \in B_{i,j}$ and $x \in B_{i,j+1}$, the neighborhood of $u$
is a strict subset of that of $x$ with respect to $B_{i+1}$ (i.e.,
$N(u) \cap B_{i+1}$ $\subsetneq$ $N(x) \cap B_{i+1}$).
That is, $  B_i$ can be partitioned into $t_i$ subsets $B_{i,1},
\ldots, B_{i, t_i}$ where: (1) if $v, v' \in B_{i,j}$, then they
have the same neighbors in $B_{i+1}$ (i.e., $N(v) \cap B_{i+1} =
N(v') \cap B_{i+1}$), and (2) if $u \in B_{i,j}$ and $v \in
B_{i,j+1}$, then  in $B_{i+1}$ the neighborhood of $u$ is a strict
subset of the of the neighborhood of $v$  (i.e., $ N(u) \cap
B_{i+1} \subsetneq N(v) \cap B_{i+1}$).
From this partitioning of $B_i$ and $B_{i+1}$ we can easily
construct arcs between $(i)$ and $(i+1)$ to capture the edges
between vertices of $B_i$ and $B_{i+1}$. This is depicted in
Figure~\ref{fig:circular_arc} and described as follows.

For every $i$, we place $t_i$ equally spaced points $\{(i,1),
\ldots, (i,t_i)\}$ on the arc from $(i)$ to $(i+1)$ and:
\begin{itemize}
\item for a vertex $b_j$ in $B_{i,j}$ we use the arc from $(i)$
to $(i,j)$.
\item for a vertex $x_j$ in $X_{i+1,j}$ we use the arc from
$(i,j)$ to $(i+1)$.
\end{itemize}
Clearly, the arc from $(i)$ to $(i,j)$ precisely intersects all
arcs from vertices in $B_i$ and the arcs of $b_j$'s neighbors in
$B_{i+1}$. Similarly, the arc from $(i,j)$ to $(i+1)$ precisely
intersects all arcs from vertices in $B_{i+1}$ and the arcs of
$x_j$'s neighbors in $B_i$.

Notice that, for each $v \in B$ there is a unique triple
$(i,j,j')$ of indices where $v \in B_{i,j} \cap X_{i,j'}$. Thus,
by performing this construction for each $i \in \{0,
\ldots,\ell-1\}$, each $v$ will be mapped to an arc from
$(i-1,j')$ to $(i,j)$; i.e., we have a circular-arc representation
for $B$.
\end{proof}
\begin{figure}[h]
\centering
\includegraphics[scale=1.2]{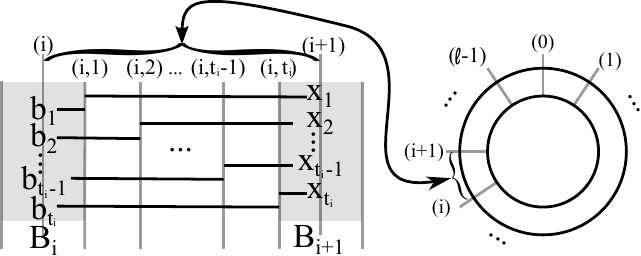}
\caption{The partial arcs between $(i)$ and $(i+1)$ where the
arc $b_j$ represents the vertices from $B_{i,j}$ and the arc
$x_j$ represents the vertices from $X_{i+1,j}$.}
\label{fig:circular_arc}
\end{figure}
Let $B$ be an $\ell$-buoy of a graph $G$ with bags $B_0, \ldots,
B_{\ell-1}$. Consider a vertex $x$ of some bag $B_i$. We say $x$
is  a {\it dominant} vertex of $B_i$ if (in $B$) it dominates
every other vertex of $B_i$.
\begin{observation}\label{obs:buoy-comparable}
Let $B$ be an odd $\ell$-buoy of an even-hole-free graph $G$ with bags
$B_0, \ldots, B_{\ell-1}$. For every $i \in
\{0,$ $\ldots,$ $\ell-1\}$ and every pair of vertices $x,y$ in
$B_i$, $x$ and $y$ are comparable in $B$. In particular, each
$B_i$ contains a dominant vertex.
\end{observation}
\begin{proof}
Figure~\ref{fig:comparable} depicts the structure we observe
in this proof.
Suppose some pair of vertices $x,y $ in $B_0$ are incomparable.
Then by Observation~\ref{obs:buoy-square}, there are
vertices $a \in B_1, b \in B_{\ell-1}$ with $xa, by \in E(G)$
and $ya, bx \not\in E(G)$. Now, by Observation~\ref{obs:hole}
and for the edges $xa$ and $by$, $B$ has skeletons
$(x, a(=a_1), a_2, \ldots, a_{\ell-1})$ and
$(y, b(=b_{\ell-1}), \ldots, b_1)$ where $a_i,b_i \in B_i$.

Notice that if $a_ib_{i+1}$ is an edge, then $(y, x, a_1, a_2,
a_3, \ldots, a_i, b_{i+1}, \ldots, b_{\ell-1})$ is an even hole.
Thus, $a_ib_{i+1} \notin E(G)$ and each $a_i$ is distinct from
each $b_j$. Moreover, $b_ia_{i+1}$ is an edge (otherwise
$a_i,b_i,b_{i+1},a_{i+1}$ is an induced $C_4$). But now $(x, a_1,
a_2, b_2, b_3, \ldots, b_{\ell-2}, a_{\ell-1})$ is an even hole
since $a_1b_2, a_2b_3 \notin E(G)$, $a_2 \neq b_2$, and
$b_{\ell-2}a_{\ell-1} \in E(G)$, a contradiction. Since the
domination relation is transitive, every bag $B_i$ contains a
vertex $d_i$ that dominates every other vertex of $B_i$, i.e.,
$d_i$ is a dominant vertex of $B_i$
\end{proof}
\begin{figure}[h]
\centering
\includegraphics[scale=1.5]{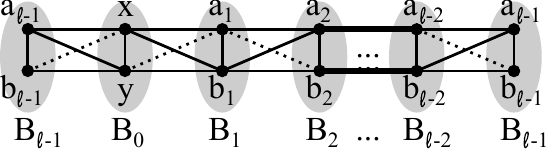}
\caption{The $\ell$-buoy from the proof of Observation~
\ref{obs:buoy-comparable}. Note: we have duplicated bag
$B_{\ell-1}$ for ease of presentation. Also, the bold edges
connecting $a_2$ with $a_{\ell-2}$ and $b_2$ with $b_{\ell-2}$
correspond the to the paths $a_2, \ldots, a_{\ell-2}$ and $b_2,
\ldots, b_{\ell-2}$ respectively.}
\label{fig:comparable}
\end{figure}
From the proof of Observation~\ref{obs:buoy-comparable}, we can
extract a linear-time algorithm to establish the following
observation.
\begin{observation}\label{obs:find-buoy-comparable}
Let $B$ be an odd $\ell$-buoy of a graph $G$ with bags $B_0, \ldots,
B_{\ell-1}$. If there are vertices $x$ and $y$ in
some $B_i$ such that $x$ and $y$ are incomparable in $B$, then $G$
contains an even hole, and this even hole can be found in linear
time. \qed
\end{observation}
Later (Lemma~\ref{lem:find-buoy-domination}) we will show that
the domination property of Observation~\ref{obs:buoy-comparable}
can be verified in linear time.
\begin{observation}\label{obs:buoy-pan}
Let $G$ be a (pan, even hole)-free graph. Let $B$ be an
$\ell$-buoy of $G$ with bags $B_0, \ldots , B_{\ell-1}$. Let $a$ and $b$
be two vertices in some $B_i$. If $a$ strictly dominates $b$ in
$B_{i+1}$, then $b$ dominates $a$ in $B_{i-1}$.
\end{observation}
\begin{proof}
Let $a$ and $b$ be two vertices in $B_i$. Suppose $a$ strictly dominates
$b$ in $B_{i+1}$, but $b$ does not dominate $a$ in $B_{i-1}$.
Thus, there are vertices $c \in B_{i+1}, d \in B_{i-1}$ such that
$ac, ad \in E(G)$ and $bc, bd \not\in E(G)$. By
Corollary~\ref{cor:hole}, there is a skeleton $C$ containing the
vertices $d,a,c$. Now $b$ together with $C$ induces a pan in
$G$.
\end{proof}

From the proof of Observation~\ref{obs:buoy-pan}, we can extract
an algorithm to establish the following observation.
\begin{observation}\label{obs:find-buoy-pan}
Let $G$ be a graph. Let $B$ be an $\ell$-buoy of $G$ with bags
$B_0, \ldots , B_{\ell-1}$. Let $a$ and $b$ be two vertices in some
$B_i$. If $a$ strictly dominates $b$ in $B_{i+1}$, and $b$ does
not dominate $a$ in $B_{i-1}$, then $G$ contains a pan or an even
hole, and such an induced subgraph can be found in linear time.
\hfill $\Box$
\end{observation}

Observations~\ref{obs:buoy-pan} and \ref{obs:buoy-comparable} tell
us that the structure of a buoy in a (pan, even hole)-free graph
is very restricted (see Corollary \ref{cor:buoy-clique} below).
Additionally, this structure allows us to prove that a buoy in a
(pan, even hole)-free graph is a unit circular-arc graph (see
Theorem \ref{thm:buoy-unit-circ} below).
\begin{corollary}\label{cor:buoy-clique}
Let $G$ be a (pan, even hole)-free graph and let $B$ be an
$\ell$-buoy of $G$ with bags $B_0, \ldots, B_{\ell-1}$. For every
$i \in \{0, \ldots, \ell-1\}$ either $B_{i-1} \cup B_i$ or $B_i
\cup B_{i+1}$ is a clique.
\end{corollary}
\begin{proof}
Consider a bag $B_i$.  By Observation~\ref{obs:buoy-comparable},
we can order the vertices of $B_i$ as $b_{i,1}, b_{i,2}, \ldots,
b_{i,t}$ such that $b_{i,k}$ dominates $b_{i,j}$ whenever $k > j$.
In particular, the vertex $b_{i,t}$ is a dominant vertex of $B_i$,
and $b_{i,t}$ is adjacent to all of $B_{i-1} \cup B_{i+1}$. We may
suppose $B_i \cup B_{i+1}$ is not a clique, for otherwise we are
done. Consider a vertex $b_{i,r}$ that is not adjacent to some
vertex $u$ in $ B_{i+1}$. Thus, $b_{i,t}$ strictly dominates
$b_{i,r}$, and therefore $b_{i,1}$. By
Observation~\ref{obs:buoy-pan}, $b_{i,1}$ dominates $b_{i,t}$ in
$B_{i-1}$, and so $b_{i,1}$ is adjacent to all of $B_{i-1}$. By
the domination order of $B_i$, every other vertex in $B_i$ is
adjacent to all of $B_{i-1}$. Thus $B_{i-1} \cup B_i$ induces a
clique.
\end{proof}

\begin{theorem}\label{thm:buoy-unit-circ}
If $B$ is an $\ell$-buoy in a (pan, even hole)-free graph $G$,
then $B$ is a unit circular-arc graph.
\end{theorem}
\begin{proof}
This proof is an easy adaptation of the construction from the
proof of Theorem~\ref{thm:buoy-square-circarc}.
By Corollary~\ref{cor:buoy-clique}, when $B_i \cup B_{i+1}$ is
not a clique, both $B_{i-1} \cup B_i$ and $B_{i+1}
\cup B_{i+2}$ must be cliques. In particular, when $B_i \cup B_{i+1}$
is not a clique we use nearly the same construction as
before and exploit the fact that $B_{i-1} \cup B_i$ and $B_{i+1}
\cup B_{i+2}$ are cliques to ensure all of our arcs have the
same length. Figure~\ref{fig:unit_circular_arc} depicts our
construction and we will refer to it as we describe the
details.

\begin{figure}[h]
\centering
\includegraphics[width=\linewidth]{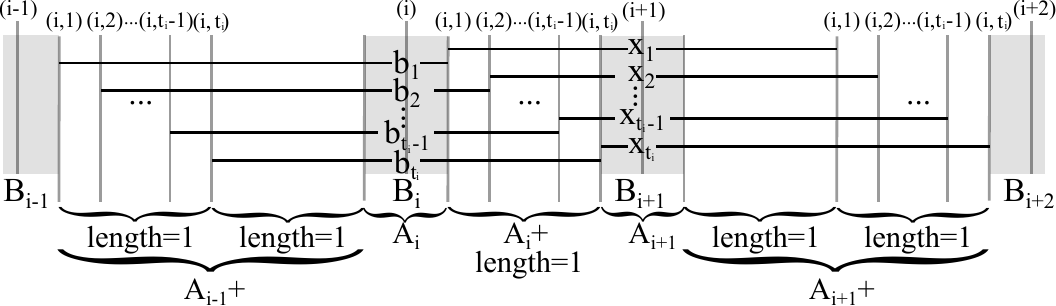}
\caption{The unit circular-arc construction from the proof of
Theorem~\ref{thm:buoy-unit-circ} for the case when $B_i \cup
B_{i+1}$ is not a clique.} \label{fig:unit_circular_arc}
\end{figure}

The arcs we construct will have length $2+\epsilon$. As in the
previous case we first partition the circle into arcs. We then
use these arcs to place the endpoints of the arcs for the
vertices of $B$.

We partition the circle into arcs as follows. For each bag
$B_i$ we allocate an arc $A_i$ of length $\epsilon$. The
midpoint of $A_i$ will be the point $(i)$ from
the proof of Theorem~\ref{thm:buoy-square-circarc}.
For each $i \in \{0, \ldots, \ell-1\}$ we allocate an arc $A_i+$ such that:
\begin{itemize}
\item When $B_{i} \cup B_{i+1}$ is a clique, the length of $A_i+$ is two.
\item When $B_{i} \cup B_{i+1}$ is not a clique, the length of $A_i+$ is one.
\end{itemize}
These arcs are arranged as $A_0,A_0+,A_1,A_1+, \ldots,
A_{\ell-1}, A_{\ell-1}+$ around the circle so that the circle
is covered and consecutive arcs intersect in precisely one
point.

As we mentioned, when $B_i \cup B_{i+1}$ is not a clique we
use the previous construction subject to the constraint that
the length of the arc from $(i,1)$ to $(i,t_i)$ is one (note:
$t_i > 1$ since $B_i \cup B_{i+1}$ is not a clique).

In the first half of $A_{i-1}+$ we insert a copy of the points
of $A_i+$.
In particular, the left endpoint of $A_{i-1}+$ is a copy of
$(i,1)$ and this is followed by $(i,2), \ldots, (i,t_i-1)$,
$(i,t_i)$ with precisely the same spacing as in $A_{i}+$.
With these points we can now create the arcs for the vertices
in $B_i$. Specifically, using the same partition $B_{i,j}$ as
before, each $b_j \in B_{i,j}$ is represented by the arc
from the copy of $(i,j)$ in $A_{i-1}+$ to the original $(i,j)$
in $A_{i}+$. It is important to note that each such arc
includes the midpoint of $A_{i-1}+$, has length $2+\epsilon$,
and includes the same points between $(i)$ and $(i+1)$ as in
our previous construction.

We similarly, insert a copy of the points of $A_i+$ in the
second half of $A_{i+1}+$. Specifically, the midpoint of
$A_{i+1}+$ is a copy of $(i,1)$, which is followed by $(i,2),
\ldots, (i,t_i-1)$, $(i,t_i)$ with precisely the same spacing
as in $A_{i}+$.
With these points we can now create the arcs for the vertices
in $B_{i+1}$. Specifically, using the partition $X_{i+1,j}$
as before, each $x_j \in X_{i+1,j}$ is represented by the arc
from the original $(i,j)$ in $A_{i}+$ to the copy of $(i,j)$
in $A_{i+1}+$. It is important to note that each such arc
includes the midpoint of $A_{i+1}+$, has length $2+\epsilon$,
and includes the same points between $(i)$ and $(i+1)$ as in
our previous construction.

We need only consider one special case to complete our
construction, namely, when both $B_i \cup B_{i+1}$ and
$B_{i} \cup B_{i-1}$ are cliques. In this case we simply map
each vertex of $B_i$ to an arc from the midpoint of $A_{i-1}+$
to the midpoint of $A_{i}+$. This again provides arcs of length
$2+\epsilon$.

Now, when $B_i \cup B_{i+1}$ is not a clique, this construction
properly represents the edges between $B_i$ and $B_{i+1}$ since we
simply have the same representation as before. Additionally, when
$B_i \cup B_{i+1}$ is a clique, the arcs of $B_i$ and $B_{i+1}$
always include the midpoint of $A_i+$. This again properly
represents the edges between $B_i$ and $B_{i+1}$. Thus, we have
produced a unit circular-arc representation of $B$.
\end{proof}
\subsection{Neighbors of Buoys}
\label{sec:buoy-neighbor}
We now generalize the results of Observation 3.1 to buoys. We
examine the different \emph{types} of adjacencies between vertices
outside a buoy $B$ in a (pan, even hole)-free graph and vertices
inside $B$. Let $x$ be a vertex of $G$ outside of $B$. We say $x$
is of type $t$ with respect to $B$ if $x$ has neighbors in exactly
$t$ distinct bags $B_i$. It is easy to see that in a pan-free
graph, $x$ cannot be of type 1. The following lemma describes
possible adjacencies between $x$ and $B$.
\begin{lemma}\label{lem:buoy-outside}
Let $G$ be a (pan, even hole)-free graph. Let $B$ be an odd $\ell$-buoy
of $G$ with bags $B_0, \ldots , B_{\ell-1}$ and let $x$ be a vertex of $G - B$ that has some neighbors in
$B$.

\begin{enumerate}[(i)]
\item If $N(x) \cap B_i \not= \emptyset$, then $N(x)\cap
B_{i-1} \not= \emptyset$, or  $N(x)\cap B_{i+1} \not= \emptyset$.
\item If $i$ and $j$ are indices such that $j \in \{2,$ $\ldots,$ $\ell
-1\}$, $x$ has a neighbor in each of $B_i$ and $B_{i+j}$, and $x$
has no neighbors in $B_{i+1} \cup$ $\ldots \cup$ $B_{i+j-1}$, then
$j$ is odd (i.e., the number of bags $B_i, \ldots, B_{i+j}$ is
even).
\item If $x$ has neighbors in $B_{i-1}$ and neighbors in
$B_{i+1}$, then $x$ is $B_i$-complete.
\item Vertex $x$ is of types 2, 3, or $\ell$. If $x$ is of type
$\ell$, then $x$ is $B$-complete.
\item If $x$ is of type 3, then $x$ has neighbors in three
consecutive $B_i$s.
\item Suppose $x$ is of type 2 and has neighbors in $B_i$ and in
$B_{i+1}$. Then $\{x\} \cup B_i \cup B_{i+1}$ is a clique.
\end{enumerate}
\end{lemma}
\begin{proof}
Let $B_i$ be a fixed bag. Let $d_t$ be a dominant vertex of $B_t$,
$0 \leq t \leq \ell -1$. The vertices $d_t$ exist for all $t$ by
Observation \ref{obs:buoy-comparable}.

\textit{Proof of (i)}. Suppose $x$ is adjacent to some $b_i \in
B_i$ and is $(B_{i-1} \cup B_{i+1})$-null. Let $b_{i-1}$
be a neighbor of $b_i$ in $B_{i-1}$ and let $b_{i+1}$
be a neighbor of $b_i$ in $B_{i+1}$. By Corollary~\ref{cor:hole}, there is a
skeleton $C$ containing $b_{i-1}, b_i, b_{i+1}$. But then part (i)
of Observation~\ref{obs:neighbors} is contradicted.  \qed

\textit{Proof of (ii)}. Suppose (ii) is false. Let $a$ and $a'$ be
neighbors of $x$ in $B_i$ and $B_{i+j}$ respectively. Since $d_t$ is
a dominant vertex of $B_t$ for all $t$,
$\{x, a,$ $d_{i+1},$ $\ldots,$ $d_{i+j-1},$ $a'\}$ induces an even hole.
\qed

\textit{Proof of (iii)}. Suppose there is $b_i \in B_i$ which is
not a neighbor of $x$. We will distinguish among three cases: (1)
$b_i$ and $x$ have common neighbors $a_{i-1} \in B_{i-1}$ and
$a_{i+1} \in B_{i+1}$; (2)  $b_i$ and $x$ have no common neighbors
in $B_{i-1} \cup B_{i+1}$; and (3)  $b_i$ and $x$ have a common
neighbor $a_{i-1} \in B_{i-1}$, but no common neighbor in
$B_{i+1}$. In case (1), $\{x, a_{i-1}, b_i, a_{i+1}\}$ induces a
$C_4$. For case (2), let $a_{i-1}$ be a neighbor of $x$ in
$B_{i-1}$ and let $a_{i+1}$ be a neighbor of $x$ in $B_{i+1}$.
Note that $b_i$ must have neighbors $b_{i-1}$ in $B_{i-1}$  and
$b_{i+1}$ in $B_{i+1}$. Thus, in
this case, $\{x, a_{i-1}, b_{i-1}, b_i, b_{i+1}, a_{i+1}\}$
induces a $C_6$.

Now we handle case (3). Let $a_{i+1}$ be a neighbor of $x$ in
$B_{i+1}$ and $b_{i+1}$ be a neighbor of $b_i$ in $B_{i+1}$. The
dominant vertex $d_i$ of $B_i$ is adjacent to both $a_{i-1}$ and
$a_{i+1}$. Thus, $d_i$ is adjacent to $x$, for otherwise $\{x,
a_{i-1}, d_i, a_{i+1} \}$ induces a $C_4$.

Suppose $x$ is not adjacent to $d_{i-2}$. By
Corollary~\ref{cor:hole}, there is a skeleton $C$ containing
$d_{i-2}, a_{i-1}, b_i$. But then $x$ and $C$ contradict
Observation~\ref{obs:neighbors} (i). Thus $x$ is  adjacent to
$d_{i-2}$. Let $P$ be the path $d_{i-2} a_{i-1} d_i b_{i+1}$.
Corollary~\ref{cor:hole} implies there is a skeleton $C$
containing $P$. Since $d_{i+2}$ is a dominant vertex of $B_{i+2}$,
we may assume $d_{i+2} \in C$ ($d_{i+2}$ may replace the vertex of
$C \cap B_{i+2}$). Since $x$ has at least three neighbors and one
non-neighbor ($b_{i+1}$) on $C$, Observation~\ref{obs:neighbors} (v)
implies that $x$ has exactly three neighbors on $C$. In
particular, we have $x d_{i+2} \not\in E(G)$. Let $C' = (C -
\{b_{i+1}\}) \cup \{a_{i+1}\}$. Then $C'$ is a hole. But now, $x$
has at least four neighbors and one non-neighbor ($d_{i+2}$) on
$C'$, a contradiction to Observation~\ref{obs:neighbors} (v). \qed

\textit{Proof of (iv)}. Suppose $x$ is of a type different from 2,
3, or $\ell$. There are indices $i,j$ such that $x$ has neighbors
in each of $B_i$ and $B_j$ and no neighbors in $B_{i+1}, \ldots,
B_{j-1}$ and $ |j - i|  \not\equiv 1$ $mod \;  \ell$. By (i), $x$
has neighbors in $B_{i-1}$ and in $B_{j+1}$. By (ii), the number
of sets $B_i, B_{i+1}, \ldots , B_j$ is even. So the number of
sets $B_{j+1} , B_{j+2}, \ldots, B_{i-1}$ is odd. Let $b_t$ be a
neighbor (if one exists) of $x$ in $B_t$, for all $t$. Consider
the path $P= b_{j+1} d_{j+2} \ldots d_{i-2} b_{i-1}$. Vertex $x$
must be adjacent to an interior vertex $p$ of this path, for
otherwise $P$ and $x$ form an even hole. But now there is a pan
formed by the vertices $p, x, b_{i}, d_{i+1}, \ldots, d_{j-1},
b_j$.

If $x$ is of type $\ell$, then by (iii), $x$ is $B_i$-complete for
all $i$. \qed

\textit{Proof of (v)} Let $x$ be of type 3. Assume $x$ has a
neighbor in $B_i$ for some $i$. By (i) we may assume $x$ has a
neighbor in $B_{i+1}$. Let $B_j$ be the third bag such that $x$
has neighbors in $B_j$. If $j \notin \{i-1, i+2\}$, then (i) is
contradicted. \qed

\textit{Proof of (vi)}. Suppose $x$ is of type 2 and has neighbors
in $B_i$ and $B_{i+1}$ and let $a_i \in B_i$ and $a_{i+1} \in
B_{i+1}$ be neighbors of $x$. If $a_i$ and $a_{i+1}$ are not
adjacent, then  $\{a_i, x,$ $a_{i+1},$ $d_{i+2},$ $\ldots,$
$d_{i-1}\}$ induces an even hole. Thus, the neighbors of $x$ in
$B$ form a clique.

We now show that $x$ is adjacent to every neighbor $b_{i+1} \in
B_{i+1}$ of $a_i$. Suppose $x$ is not adjacent to $b_{i+1}$.
Consider the skeleton $C$ formed by the vertices $a_i, b_{i+1},
d_{i+2}, d_{i+3}, \ldots, d_{i-1}$. Vertex $x$ has only one
neighbor on this hole, a contradiction to part (v) of
Observation~\ref{obs:neighbors}. By symmetry, $x$ is adjacent to
every neighbor $b_i \in B_i$ of $a_{i+1}$.

From the previous paragraph $x$ must be adjacent to both $d_i$ and
$d_{i+1}$ since they are neighbors  of $a_{i+1}$ and $a_i$
respectively. Thus, $x$ is adjacent to all of $B_i \cup B_{i+1}$
and as such $\{x\} \cup B_i \cup B_{i+1}$ form a clique.
\end{proof}
From the proof of Lemma~\ref{lem:buoy-outside}, we can extract a
linear-time algorithm to establish the following lemma.
\begin{lemma}\label{lem:find-buoy-outside}
Let $G$ be a graph. Let $B$ be an odd $\ell$-buoy of $G$ with bags
$B_0, \ldots , B_{\ell-1}$, and let $x$ be a
vertex of $G - B$ that has some neighbors in $B$. If $x$ fails to
satisfy (i)--(vi) of Lemma~\ref{lem:buoy-outside}, then $G$
contains a pan or an even hole, and such an induced subgraph can
be found in linear time. \hfill $\Box$
\end{lemma}

\subsection{Structure Theorem}
\label{sec:structure}

Now that we understand the structure of buoys (see Section
\ref{sec:buoy}) and their neighbors (see Section
\ref{sec:buoy-neighbor}), we are ready to prove the structure
theorem introduced in Section~\ref{sec:definitions}. We prove the
following  theorem which together with
Theorem~\ref{thm:buoy-unit-circ} implies
Theorem~\ref{thm:structure}. Recall that ${\cal C}$ is the class
of graphs $G$ such that every atom of $G$ is (pan, even
hole)-free.
\begin{theorem}\label{thm:structure_buoy}
If $G$ is a connected graph in ${\cal C}$ then
\begin{enumerate}[(i)]
 \item $G$ is a clique, or
 \item $G$ contains a clique cutset, or
 \item For every maximal buoy $B$ of $G$, either $B$ is a full
 buoy of $G$, or $G$ is the join of $B$ and a clique.
\end{enumerate}
\end{theorem}
\begin{proof}
We may assume $G$ is connected and contains an odd hole $C$, for
otherwise $G$ is chordal and the theorem holds. Let $\ell$ be the
length of $C$. Since $G$ contains $C$, $G$ contains a maximal buoy
$B$ with bags $B_0, B_1, \ldots , B_{\ell -1}$ and skeleton $C$
(here, as usual, ``\emph{maximal}'' is meant with respect to set
inclusion, not size). If $G - B = \emptyset$, then $G$ is a full
buoy and we are done. Let $A$ be the set of vertices in $G - B$
with some neighbor in $B$, and $R$ be the set of vertices in $G$
with no neighbor in $B$. Consider a vertex $x$ in $A$. By
Lemma~\ref{lem:buoy-outside}, $x$ is of types 2, 3, or $\ell$. If
$x$ is of type 3, then, by $(iii)$ and $(iv)$ of
Lemma~\ref{lem:buoy-outside}, $x$ has neighbors in three
consecutive bags $B_{i-1}, B_i, B_{i+1}$ and is complete to $B_i$.
In particular, $B \cup \{x\}$ is a larger buoy with bags $B_0,
B_1, \ldots , B_{i-1}, B_{i} \cup \{x\}, B_{i+1}, \ldots, B_{\ell
-1}$, a contradiction to our choice of $B$. Also, when $x$ is of
type 2, then, by $(ii)$ and  $(vi)$ of
Lemma~\ref{lem:buoy-outside}, there is an index $i$ such that
$\{x\} \cup B_i \cup B_{i+1}$ is a clique. Thus, $A$ can be
partitioned into sets $A_0, A_1, \ldots A_{\ell-1}, U$ such that
\begin{itemize}
\item $a \in A_i$ if and only if $N(a) \cap B = B_i \cup B_{i+1}$;
and \item $u \in U$ if and only if $u$ is $B$-complete.
\end{itemize}

Note that all type-$\ell$ vertices are in $U$. The set $U$ (if
non-empty) induces a clique for otherwise,  two non-adjacent
vertices of $U$ and  two non-adjacent vertices of $B$ form a
$C_4$. We may assume there is a non-empty $A_i$, for otherwise $G$
is the join of $B$ and $U$ (if $R= \emptyset$), or  $U$ is a
clique cutset of $G$ separating $B$ and $R$ (if $R \not=
\emptyset$).

Consider a non-empty set $A_i$ and let $D = B_i \cup B_{i+1} \cup
U$. The set $D$ is a clique by Lemma~\ref{lem:buoy-outside}. We
will show that $D$ is a clique cutset. Suppose it is not a clique
cutset. Then, in $G-D$, there is a shortest path $P$ from a vertex
$a_i \in A_i$ to a vertex $b \in B - (B_i \cup B_{i+1})$.
Enumerate the vertices of $P$ as $v_1, v_2, \ldots, v_t$ with $v_1
= a_i$ and $v_t = b$. Since the path is shortest, $v_{t-1} \in
A_j$ for some $j \not=i$, $t \geq 3$,  and $v_j \in R$ for $j \in
\{2,3, \ldots t-2\}$ (when $t > 3$). There are two induced paths
whose endpoints are $v_{t-1}$ and $a_i$, and whose interior
vertices are disjoint and lie in $B$. We may enumerate one as $P_1
= v_{t-1} b_{j+1} \ldots b_i a_i$, and the second one as $P_2 =
v_{t-1} b_j \ldots b_{i+1} a_i$ with $b_k \in B_k$ for all $k$.
Since $\ell$ is odd,  $P_1$ and $P_2$ have different parities. Let
$P' = P - \{v_t\}$. One of the two holes induced by $P_1 \cup P'$
and $P_2 \cup P'$ has to be even, a contradiction.
\end{proof}
An algorithm can be extracted from the proof of
Theorem~\ref{thm:structure_buoy} to prove the following theorem.
\begin{theorem}\label{thm:find-structure-buoy}
Let $B$ be a maximal buoy of a graph $G$. If $B$ is not a full buoy of $G$ and
if $G$ is not the join of B and a clique, then $G$ contains an even hole or a pan,
and such an induced subgraph can be found in linear time. \hfill $\Box$
\end{theorem}

\section{A coloring algorithm for (pan, even hole)-free graphs}
\label{sec:coloring}
In this section, we discuss a polynomial-time algorithm to color a
graph in ${\cal C}$. Consider a graph $G$  with a clique cutset
decomposition tree $T(G)$. From the discussion in
Section~\ref{sec:definitions}, if we can color the atoms of $G$ in
polynomial time, then we can also color $G$. The purpose of this
section is to show that $G$ can indeed be colored in polynomial
time.

In \cite{OrlBon1991}, an $O(n^2)$-time algorithm is given for coloring
proper circular-arc graphs. For unit circular-arc graphs, this can be improved.
First, we use the $O(n+m)$-time algorithm of \cite{LS2008} for recognizing unit circular-arc
graphs to construct a unit circular-arc representation. Then we use the $O(n^{1.5})$-time
algorithm of \cite{ShihHsu} to find a minimum coloring of a unit circular-arc graph
given the representation.  This gives an $O(n^{1.5}+m)$-time algorithm to color unit
circular-arc graphs.

Thus, from
Theorems~\ref{thm:buoy-unit-circ} and \ref{thm:structure_buoy}, we have the following two results.
\begin{theorem}\label{thm:color-join}
There is an $O(n^{1.5}+m)$-time algorithm to find a minimum coloring of
a (pan, even hole)-free graph that is either a buoy or the join of  a buoy and a clique.
\end{theorem}
\begin{proof}
Let $G$ be a (pan, even hole)-free graph that is the join of a
clique $K$ and a buoy $B$. Then we have $\chi(G) = |K| + \chi(B)$.
Thus, we only need to establish the theorem for (pan, even
hole)-free buoys. Now the result follows from
Theorem~\ref{thm:buoy-unit-circ} and the $O(n^{1.5}+m)$-time
algorithm to color unit circular-arc graphs.
\end{proof}

\begin{theorem}\label{thm:color-main}
There is an $O(n^{2.5}+nm)$-time algorithm to find a minimum coloring of
a graph in ${\cal C}$.
\end{theorem}
\begin{proof}
By the discussion above and the fact that the clique cutset
decomposition provides at most $n-1$ atoms, we only need show
there is an $O(n^{1.5}+m)$-time algorithm to color a  (pan, even
hole)-free atom $G$. By Theorem~\ref{thm:structure}, $G$ is one of
the following: a clique, a buoy, or the join of a clique and a
buoy. Thus, by Theorem~\ref{thm:color-join}, $G$ can be optimally
colored in $O(n^{1.5}+m)$ time.
\end{proof}

\section{Recognition algorithms for (pan, even hole)-free graphs}
\label{sec:recognition}
In this section, we give two polynomial-time algorithms to recognize
(pan, even hole)-free graphs. We note that a polynomial-time algorithm
for recognizing (pan, even hole)-free graphs can easily be
converted to a polynomial-time algorithm for recognizing graphs in~${\cal C}$.

There exist several polynomial-time algorithms (\cite{ChaLu2013,
ChuKaw2005, ConCor2002}) for finding an even hole in a graph. But
the fastest such algorithm \cite{ChaLu2013} runs in time
$O(n^5m^3) \le O(n^{11})$. A straight-forward algorithm to
recognize a (pan, even hole)-free graph is to test for a pan using
Theorem~\ref{thm:find-pan} below, and then to test for an even
hole. In particular, we can recognize (pan, even hole)-free graphs
in $O(n^5m^3)$ time. We will design faster algorithms for (pan,
even hole)-free graph recognition. We provide two recognition
algorithms. The first uses the fact that the (pan, even hole)-free
atoms are unit circular-arc graphs and recognizes (pan, even
hole)-free graphs in $O(n m^2 + n^2 m \log\log n)$ time. The
second uses the fact that the atoms are essentially very special
buoys and runs in $O(nm)$ time.

Similar to our coloring algorithm, we note that detecting an
even hole in a graph $G$ is easily reduced to checking for an
even hole in an atom. That is, suppose a graph $G$ has a clique
cutset $C$ and consider the subgraphs $G_1 = G[V_1]$ and $G_2 =
G[V_2]$ where $V = V_1 \cup V_2$ and $C = V_1 \cap V_2$. Then $G$
contains an even hole if and only if $G_1$ or $G_2$ does; i.e.,
when testing for even holes one need only consider atoms.

As we have mentioned previously, the clique cutset decomposition
tree $T(G)$ can be computed in $O(nm)$ time such that there are
fewer than $n$ atoms \cite{Tar1985}.

\subsection{Recognition via testing for pans and then testing for even holes in unit circular-arc atoms}

We will first describe an algorithm to find a pan in
a graph.

\begin{lemma}\label{lem:hole-in-atom}
Let a graph $G$ be an atom. Then every vertex $v$ of $G$ is
universal, or lies in a hole of $G$. Furthermore, there is a
linear-time algorithm to find a hole containing $v$ when $v$ is
not universal.
\end{lemma}
\begin{proof}
Let $G$ be an atom and $v$ be a vertex of $G$. Let $M(v) = V(G) -
(N(v) \cup \{v\})$. If $M(v ) = \emptyset$, then $v$ is universal.
Compute the components $C_1, \ldots , C_t$ of $G[M(v)]$. For each
$C_i$, compute the set $N_i$ of vertices in $N(v)$ that have some
neighbors in $C_i$. If some $N_i$ is a clique, then $N_i$ is a
clique cutset separating $v$ from $C_i$, a contradiction. Thus,
none of the $N_i$s are cliques. Choose an arbitrary set
$N_i$. Consider two non-adjacent vertices $x$ and $y$ in $N_i$. Find a
chordless path $P$ from $x$ to $y$ whose interior vertices lie
entirely in $C_i$. Then $P$ and $v$ induce a hole in $G$.
\end{proof}
\begin{lemma}\label{lem:find-hole}
Given a graph $G$ and a vertex $v$ in $G$, there is an $O(nm)$
time algorithm to find a hole containing $v$, if such a hole
exists.
\end{lemma}
\begin{proof}
Construct in $O(nm)$ time the clique cutset decomposition $T(G)$
of $G$. Consider all the atoms of $T(G)$ containing $v$. If $v$ is a universal
vertex in all such atoms, then $v$ does not lie on any hole of $G$.
Suppose $v$ is not universal in some atom $A$. By
Lemma~\ref{lem:hole-in-atom}, we can find a hole
containing $v$ in linear time.
\end{proof}
\begin{theorem}\label{thm:find-pan}
There is an $O(n m^2)$-time algorithm to find a pan in a graph, if
one exists.
\end{theorem}
\begin{proof}
For an edge $ab$ we can check, by Lemma~\ref{lem:find-hole}, in
$O(nm)$ time whether $ab$ is the handle of a pan by finding a
hole containing $a$ (respectively, $b$) in the subgraph of $G$
induced by $V(G) - (N(b) - \{a\})$ (respectively, $V(G) - (N(a) -
\{b\})$.) Since $G$ has $m$ edges, the time bound of the theorem
follows.
\end{proof}

Now, to recognize whether $H$ is a (pan, even hole)-free graph, we
first use Theorem~\ref{thm:find-pan} to test for a pan. If $H$ has
no pan, find the clique cutset decomposition.

For an atom $G$, Theorem~\ref{thm:structure} implies that $G$ is
either a unit circular-arc graph or the join of a clique $K$ and a
unit circular-arc graph $G'$. In the latter case,  $G$ is
even-hole-free if and only if $G'$ is even-hole-free. One can test
whether a graph is a unit circular-arc graph in linear time
\cite{LS2008}. In particular, if $G'$ is not unit circular-arc,
then we know $G$ must have an even hole (by
Theorem~\ref{thm:buoy-unit-circ}). Additionally, an $O(n m  \log
\log n)$-time algorithm is known for finding an even (or odd) hole
in a circular-arc graph \cite{CamEsc2007}. That is, via the clique
cutset decomposition, we can test whether a graph in ${\cal C}$
contains an even hole in $O(n^2 m \log \log n)$ time (since the
decomposition can be computed in $O(n m)$ time and has $O(n)$
atoms).

Thus, for a given graph $H$, we can recognize whether $H$ is
(pan, even hole)-free in $O(n m^2 + n^2 m \log \log n)$ time.

\subsection{Recognition via buoy construction}

We now present an algorithm which relies on the buoy structure of
a (pan, even hole)-free graph to test whether an atom is (pan,
even hole)-free, and if it is not, to find a pan or even hole.
Recall that, by Theorem \ref{thm:structure_buoy}, in a (pan, even
hole)-free atom $G$ either every maximal buoy is a full buoy or
$G$ is the join of a clique and a buoy. With this approach, we do
not attempt to directly find a pan. Instead, a pan (if it exists)
can be found by examining the buoys and their neighborhoods.

An atom $A$ of graph $G$ is {\it maximal} if any induced subgraph
$H$ of $G$ that properly contains $A$ is not an atom, i.e., if $H$
has a clique cutset. The atoms produced by the clique cutset
decomposition are maximal.

The algorithm will produce a forbidden induced subgraph, if one
exists. The algorithm has three steps.
\begin{enumerate}[(1)]
 \item Find a clique cutset decomposition tree $T(G)$ of $G$.
 \item For each (maximal) atom $A$ of $T(G)$,  (i) extract a forbidden induced
subgraph (if one exists) from $A$, or (ii) show that $A$ is a buoy, or (iii)  find a partition of
the vertices of $A$ into the
join of a buoy and a clique. The involved buoy will
satisfy Observation~\ref{obs:buoy-comparable}.
 \item For each atom $A$ of $T(G)$, verify that no holes of $A$
    form a pan with a vertex outside~$A$.
\end{enumerate}
We will show that steps (2) and (3) can be done in linear time for an
atom. This shows the algorithm runs in $O(n m)$ time.

The correctness of step (2) follows from the following theorem.
\begin{theorem}\label{thm:test-atom}
Let $G$ be an atom. There is a linear-time algorithm to output
\begin{enumerate}[(i)]
   \item a pan, or
   \item an even hole, or
   \item a certificate that $G$ is (pan,
   even hole)-free, and  either a certificate that $G$ is a buoy or
   a partition of $V(G)$ into sets $B$ and
   $K$  such that $B$ is a buoy,  $K$ is a clique, and $G$ is the
   join of $B$ and $K$.
\end{enumerate}
\end{theorem}
To prove Theorem~\ref{thm:test-atom}, we will need the following
three lemmas.
\begin{lemma}\label{lem:buoy-length}
If $B$ is an $\ell$-buoy where each $B_i$ can be ordered by
neighborhood inclusion, then every hole in $B$ has length $\ell$.
\end{lemma}
\begin{proof}
Consider a hole $H$ of $B$. No two vertices of $H$ have comparable
neighborhoods. Thus, each bag of $B$ contains at most one vertex
of $H$. If $H$ has fewer than $\ell$ vertices, then  $H$ is not a
hole (a contradiction) since vertices in a bag can only have
neighbors in the bag preceding it and the bag following it in the
cyclic order. So $H$ has length $\ell$.
\end{proof}
\begin{lemma}\label{lem:clique}
Let $B$ be an odd $\ell$-buoy where each bag
$B_i$ can be ordered by neighborhood inclusion. The following
three statements are equivalent
\begin{enumerate}[(i)]
 \item There are  two vertices $a, b$ in some $B_i$ such that
       $a$ strictly dominates $b$ in $B_{i-1}$, but $b$ does not
       dominate $a$ in $B_{i+1}$.
 \item $B$ has a pan.
 \item There is a subscript $i$ such that $B_{i-1} \cup B_i$ and
       $B_i \cup B_{i+1}$ are
        both not cliques.

\end{enumerate}
\end{lemma}
\begin{proof}
First, note that Lemma~\ref{lem:buoy-length} implies that $B$ has no
even hole.  Now the fact that  (i) $\Longrightarrow$ (ii) follows
from Observation~\ref{obs:find-buoy-pan}. Next, we prove the
implication (ii) $\Longrightarrow$ (iii). Suppose $B$ contains a
pan. By Lemma~\ref{lem:buoy-length}, the hole of this pan must
contain exactly one vertex of each bag. Let the pan consist of
vertices $a, b_0, b_1, \ldots, b_{\ell -1}$ where $b_i \in B_i$
and the vertices $b_i$ form a hole. Without loss of generality, we
may assume $a \in B_1$. So, (iii) is satisfied with $i =1$.
Finally, we prove the implication (iii) $\Longrightarrow$ (i).
Suppose $B_{0} \cup B_1$ and $B_1 \cup B_{2}$ are both not
cliques. Let $a$ be the dominant vertex of $B_1$, and $b$ be the
vertex in $B_1$ that is dominated by every other vertex of $B_1$.
If $b$ is adjacent to every vertex in $B_0$, then every vertex in
$B_1$ is adjacent to every vertex in $B_0$, a contradiction. So
$b$ is non-adjacent to some vertex of $B_0$, i.e., $a$ strictly
dominates $b$ in $B_0$. A symmetric argument shows that $a$ strictly
dominates $b$ in $B_2$.
\end{proof}
\begin{lemma}\label{lem:find-buoy-domination}
Let $B$ be an $\ell$-buoy  with bags $B_0, \ldots, B_{\ell-1}$ for
some $\ell$. There is a linear-time algorithm to verify that the
bags of $B$ admit a domination order, i.e., the vertices of each
$B_i$ are pairwise comparable.
\end{lemma}
\begin{proof}
For each bag $B_i $, we order its vertices by non-decreasing size
of their neighborhoods; i.e., $B_i$ is ordered as $b_i^{0},
\ldots, b_i^{k_i-1}$, where $|N(b_i^0)| \leq |N(b_i^1)| \leq
\ldots \leq |N(b_i^{k_i-1})|$ with $k_i = | B_i |$. This can be
done in $O(|B_{i-1} |+|B_i |+|B_{i+1} |)$ time via bucket-sort.
That is, sorting all of the bags can be done in $O(n)$ time. We
then check that for every $j \in \{1, \ldots, k_i -1\}$, every
neighbor of $b_i^{j-1}$ is a neighbor of $b_i^j$ (if this is not
the case, then $b_i^{j-1}$ is incomparable with $b_i^j$). For each
bag $B_i$, this neighborhood checking can be performed in
$O(\Sigma^{k_i-1}_{j=0} |N(b_i^j)|)$ time. In particular, all such
checking can be performed in $O(m)$ time. Thus, we can check that
the bags of $B$ admit a domination order in $O(n+m)$ time.
\end{proof}

Now we can prove Theorem~\ref{thm:test-atom}.

\begin{proof}
Suppose that $G$ is an atom.  Using the linear-time algorithm in
\cite{TarYan1984}, we either confirm that $G$ is chordal (and
hence is (pan, even hole)-free) or obtain a hole $H$. We may
assume $H$ is an odd hole. We first briefly describe the
algorithm. We will construct a maximal buoy $B$ with $H$ as its
skeleton in $O(n+m)$ time. During this process, we verify that $B$
has the domination property of
Observation~\ref{obs:buoy-comparable} or $G$ contains a pan or
even hole. If $H \not= G$ and $G$ is not the join of $B$ and a
clique, then we will find a pan or even hole.

We now describe our construction of a maximal $\ell$-buoy $B$ in
an atom $G$ from a hole $H = \{h_0, \ldots, h_{\ell-1}\}$ of $G$.
We start with the initial $\ell$-buoy $B$ with bags $B_0 =
\{h_0\}$, $\ldots$, $B_{\ell-1} = \{h_{\ell-1}\}$.

Let $K$ be the set of universal vertices of $G$. We remove
vertices in $K$ from $G$ since these cannot be part of an even
hole or a pan. Since $K$ is not a clique cutset of $G$, removing
$K$ does not make the resulting graph disconnected. Also, if $C$
is a clique cutset of $G_K$, then $C \cup K$ is a clique cutset of
$G$. Thus, the graph we obtain by removing $K$ is still an atom.

Consider a vertex $x$ in $G - B$ with some neighbors in $B$. If
$x$ is of type $t$ with $t \not\in \{2,3, \ell \}$, then by
Lemma~\ref{lem:buoy-outside}, we know $G$ has a pan or even hole,
and we can find this forbidden induced subgraph by
Lemma~\ref{lem:find-buoy-outside}. So, $x$ is of type 2, 3 or
$\ell$. The only candidates to be added to $B$ are type 3
vertices.  Suppose $x$ is of type 3. By
Lemma~\ref{lem:buoy-outside}, either $x$ is adjacent to three
consecutive bags $B_{i-1}, B_i, B_{i+1}$ of $B$, or $G$ contains a
pan or an even hole. Suppose $x$ is adjacent to three consecutive
bags $B_{i-1}, B_i, B_{i+1}$. If $x$ is not adjacent to all
vertices of $B_i$, then by Lemmas~\ref{lem:buoy-outside}
and~\ref{lem:find-buoy-outside}, we will find a forbidden induced
subgraph. Now $x$ is adjacent to all vertices of $B_i$. We then
add $x$ to $B_i$. We summarize the operations described in the
above paragraph with Algorithm 1, named ENGLARGE and given below.
\begin{algorithm}
\caption{ENLARGE} \label{alg:enlarge}
\begin{algorithmic}

    \STATE
    \STATE Iterating over the edges from $B$ to $G - B$,  label
    each vertex of $G - B$ with the bags of its
neighbors in $B$.
    \FOR{every vertex $x$ of $G - B$ with a label}
       \IF {$x$ has $t$ labels with $t \not\in \{2,3, \ell\}$}
             \STATE output a pan or even hole, and stop
       \ENDIF
       \IF{$x$ is labelled with three non-consecutive indices}
             \STATE   output a pan or even hole, and stop
       \ENDIF
       \IF {$x$ is labelled with three consecutive indices (say, $B_{i-1}, B_i, B_{i+1}$)}
            \IF{$x$ is not adjacent to all of $B_i$}
                       \STATE output a pan or even hole, and stop
            \ELSE
            \STATE add $x$ to $B_i$
            \ENDIF
        \ENDIF
    \ENDFOR

\end{algorithmic}
\end{algorithm}

Starting with our first buoy $B$ which is an odd hole, we call
ENLARGE on $B$ twice. We will show that after two calls to
ENLARGE, we can decide whether $G$ is (pan, even hole)-free. After
the first (respectively, second) call to ENLARGE, let $B^1$
(respectively, $B^2$) be the resulting buoy, and let the bags of
$B^1$  (respectively, $B^2$) be $B_0^1, B_1^1, \ldots, B_{\ell
-1}^1$ (respectively, $B_0^2, B_1^2, \ldots, B_{\ell -1}^2$). Note
that $B_i^1 \subseteq B_i^2$ for all $i$. Using
Lemma~\ref{lem:find-buoy-domination}, we verify in linear time
that both $B^1$ and $B^2$ have the desired domination property, or
else we find a pan or even hole.

Suppose there is a vertex of $G$ not belonging to any $B_i^2$.
Consider a vertex $x$ in $G - B^2$ with neighbors in some of the
bags. If $x$ is of type $t$ with $t \not\in \{2,3,\ell \}$, then
by Lemma~\ref{lem:find-buoy-outside}, we can produce a pan or even
hole in linear time.

We will prove that $x$ is of type 2 or $\ell$. Suppose $x$ is of
type 3. If $x$ does not have neighbors in three consecutive bags,
then by Lemma~\ref{lem:find-buoy-outside}, we can produce a pan or
even hole in linear time. So $x$ has neighbors in three
consecutive bags, say, $B_{i-1}^2, B_i^2, B_{i+1}^2$. Vertex $x$
is adjacent to all of  $B_i^2$, for otherwise, by
Observation~\ref{lem:find-buoy-outside} we can find a pan or even
hole. So, $B^3 = B^2 \cup \{x\}$ is a buoy with bags $B_0^2,
\ldots, B_{i-1}^2, B_i^2 \cup \{x\},  B_{i+1}^2, \ldots, B_{\ell
-1}^2$ . We may assume $B^3$ has the domination property of
Observation~\ref{obs:buoy-comparable}, for otherwise by
Observation~\ref{obs:find-buoy-comparable}, we will find a pan or
even hole. Thus, each bag $B_j^2$ of $B^3$ has a dominant vertex
$d_j$. Since the vertex $d_j$ is adjacent to three vertices of
$H$, $d_j$ is added to the buoy $B^1$ in the first iteration.
Vertex $x$ is adjacent to $d_{i-1}, h_i, d_{i+1}$, so $x$ would have
been added to $B^2$ in the second iteration, a contradiction.

So $x$ is of type 2 or $\ell$. (From now on, we only refer to the
buoy produced after the second call; so to simplify notation, we
will let $B=B^2$.) When $x$ is of type 2, then, by  $(vi)$ of
Lemma~\ref{lem:buoy-outside}, there is an index $i$ such that
$\{x\} \cup B_i \cup B_{i+1}$ is a clique, or else we can produce
a pan or even hole.   Thus, by $(iv)$ and $(vi)$ of
Lemma~\ref{lem:buoy-outside}, $G - B $ can be partitioned into
sets $A_0, A_1, \ldots A_{\ell-1}, U, R$ such that
\begin{itemize}
 \item $a \in A_i$ if and only if $N(a) \cap B  = B_i  \cup B_{i+1}
$,
 \item $u \in U$ if and only if $u$ is $B $-complete,
 \item $r \in R$ if and only if $r$ is $B$-null.
\end{itemize}
We may now use the proof of Theorem~\ref{thm:structure_buoy} to
find a pan or even hole. The set $U$ (if non-empty) induces a
clique for otherwise,  two non-adjacent vertices of $U$ and two
non-adjacent vertices of $B$ form a $C_4$. If all sets $A_i$ are
empty, then $U$ is a clique cutset of $G$. So, some $A_i$ is
non-empty. Since $G$ is an atom, $U \cup B_i \cup B_{i+1} $ is not
a clique cutset separating $A_i$ from $B  - (B_i \cup B_{i+1})$.
Thus, there is a shortest path $P$ from a vertex $a_i \in A_i$ to
a vertex $a_j \in A_j$ ($i \not= j$) whose interior vertices lie
entirely in $G - (B \cup U)$. Find an induced path $P'$ with the
same parity as $P$ from $a_i$ to $a_j$ whose interior vertices
belongs to $B $. Then $P \cup P'$ is an even hole.

Thus, after two calls to ENLARGE, we have constructed a full buoy
$B $ of the graph $G - K$, where $K$ is the set of universal
vertices we remove before the first call to ENLARGE.

To complete the proof, we only need to find a pan in $B $, if one
exists. By Lemma~\ref{lem:clique}, $B $ has no pan if and only if
for every $i$, $B_{i-1} \cup B_i$ or $B_i \cup B_{i+1}$ is a
clique. This condition can be checked in $O(m)$ time. If the
condition fails for some $i$, then the proof of
Lemma~\ref{lem:clique} shows that the dominant vertex $a$ of $B_i$
strictly dominates the vertex $b \in B_i$ with the smallest
degree; and so we can find a pan using
Observation~\ref{obs:find-buoy-pan}. If $K$ (the set of universal
vertices of $G$) is non-empty, then $G$ is the join of the buoy $B
$ and $K$; otherwise, $B $ is a full buoy of $G$.
\end{proof}

Now we show that step (3) of our algorithm can be implemented in
$O(nm)$ time. At this point, we know that the (maximal) atom $A$
of $T(G)$ under consideration is (pan, even hole)-free, and that
$A$ is either a buoy or the join of a buoy and a clique. We need
to determine that no hole of $A$ forms a pan with a vertex in $G -
A$; we call such a pan {\em straddling}. We need to find a
straddling pan with respect to $A$ (if one exists). An atom $A$ of
a graph $G$ is {\it maximal} if for any induced subgraph $H$ of
$G$ containing $A$, either $H$ has a clique cutset, or $A$ is a
component of $H$. The atoms produced by the clique cutset
decomposition are maximal. We will need the following observation.
\begin{observation}\label{obs:maximal-atom}
If $A$ is a maximal atom of a graph $G$, then for every vertex $x$
in $G - A$, $N_A(x)$ is either empty or a clique.
\end{observation}
\begin{proof}
Suppose $N_A(x)$ is not empty but is not a clique. Write $G' = G[A
\cup \{x\}]$. Since $A$ is a maximal atom, $G'$ is not an atom,
i.e., $G'$ has a clique cutset $C$. If $x \in C$, then $C - \{x\}$
is a clique cutset of $A$, a contradiction. So, $x$ belongs to a
component $P$ of $G' - C$. Vertex $x$ cannot be the only vertex of
$P$, for otherwise, $ N_A(x)$ is a subset of $C$, and therefore a
clique, a contradiction. But now $C$ is a clique cutset of $A$, a
contradiction.
\end{proof}
Now consider
an atom $A$ of $T(G)$ that is either a buoy $B$
or the join of a
buoy $B$ and a clique $K$. We are going to
describe a way to  find a straddling pan (if one exists) whose
hole belongs to $B$. (Vertices of $K$ do not belong to a hole in $A$.)
Let the bags of $B$ be $B_0, B_1, \dots, B_{\ell -1}$.
Remember that $A$ is (pan, even hole)-free, so by Observation~\ref{obs:buoy-comparable},
each bag can be ordered by neighborhood inclusion, so by Lemma~~\ref{lem:buoy-length},
every hole in $B$ has length $\ell$.
Compute the set $Q$ of vertices of $G - A$
that have neighbors in $B$. The set $Q$ can be computed in $O(m)$
time. Let $Q_i$ be the set of vertices $x$ of $Q$ with $N_B(x)
\cap B = B_i \cup B_{i+1}$. Since $A$ is a maximal atom, the graph
$G[A \cup Q]$ contains a clique cutset $C$ such that $C \subset
A$. It follows from Observation~\ref{obs:maximal-atom} that, with
respect to the buoy $B$, every vertex in $Q$ is of type 1 or 2.
Furthermore, since $A$ is (pan, even hole)-free, it follows from Lemma~\ref{lem:buoy-outside}(vi) 
that every vertex of type 2 belongs to some $Q_i$. If
some $x \in Q$ is of type 1, then clearly a straddling pan can be
found in linear time.
Now, every vertex $x$ in $Q_i$ is such that  $\{x\} \cup B_i
\cup B_{i+1}$ is a clique, and we conclude there is no straddling
pan.

Thus, for a maximal atom, we can determine in $O(m)$ time whether
a straddling pan exists. Since there are at most $n-1$ atoms of
$T(G)$, we can implement step 3 in $O(nm)$ time. This completes the proof
of Theorem \ref{MainResult}:
\\
\\
\noindent\textbf{Theorem \ref{MainResult}} \textit{Given a graph $G$, a pan or even hole of $G$, if one exists, can be found in $O(nm)$ time.}
\\
Note that for an input graph $G$, if $G$ is not (pan, even hole)-free, our algorithm produces a pan or an even hole.  If $G$ is (pan, even hole)-free, the algorithm produces a clique cutset decomposition tree which satisfies Theorem \ref{thm:certifying} below; furthermore, the set of atoms of every clique cutset decomposition tree will satisfy (i) and (ii) below.

\begin{theorem}\label{thm:certifying}
A graph $G$ is (pan, even hole)-free if and only if there is a clique cutset decomposition tree with most $n-1$ atoms $G_j$ such that
\begin{itemize}
\item[(i)] Each atom $G_j$ is either a clique or consists of a buoy $B(G_j)$ and a possibly empty set $U_j$ of universal vertices; the buoy $B(G_j)$ has an odd number of bags; each bag can be ordered by neighborhood inclusion; and, for each consecutive triple of bags either the first two or second two form a clique.
\item[(ii)] Further, for each atom $G_j$ which is not a clique, the neighborhood of $V(G_j)$ in $G$ can be partitioned into sets $A_i$, some of which may be empty, where $A_i$ is universal to the $i$th and $(i+1)$st bags of the buoy $B(G_j)$ and $A_i$ has no other neighbors in~$B(G_j)$.
\end{itemize}
\end{theorem}

The correctness of the algorithm proves the ``only if" part of the theorem. To see that the ``if" part holds, note that a graph $G_j$ satisfying property (i) of Theorem \ref{thm:certifying} is (pan, even hole)-free.  Since any hole of $G$ must lie in some atom $G_j$, property (ii) then guarantees that there is no straddling pan whose hole is in $G_j$.

It follows from Theorem \ref{thm:certifying} that our algorithm is certifying.  The certificate given by Theorem \ref{thm:certifying} has size $O(nm)$.

\section{Tree-width and $\chi$-boundedness}
\label{sec:tw}

In this section we bound the \emph{tree-width} of (pan,
even hole)-free graphs in terms of their clique number (see
Theorem \ref{thm:tw}). This bound immediately provides a bound on
the chromatic number (see Corollary~ \ref{cor:chi-boundedness}).

A \emph{tree decomposition} $(T,t)$ of a graph $G$ is defined to
be a tree $T$ together with a function $t: V(G) \rightarrow
V(T)$ such that:
\begin{itemize}
\item For every $v \in V(G)$, $t(v)$ induces a subtree of $T$.
\item For every $uv \in E(G)$, $t(v) \cap t(u) \neq \emptyset$.
\end{itemize}
Notice that a simple tree decomposition of any graph can be
obtained by choosing $T$ to be a single vertex and mapping every
vertex of $G$ to this vertex. The vertices of $T$ are often
treated as sets, referred to as the \emph{bags} of the tree
decomposition, and the elements of a bag $b$ are the vertices $v$
of $G$ where $b \in t(v)$. For a tree decomposition $(T,t)$ of a
graph $G$, the \emph{tree-width} of $(T,t)$, denoted $tw(T,t)$, is
the size of the largest bag of $T$ minus 1; i.e., $tw(T,t) =
\max_{b \in T} (|b| -1$). For a graph $G$, the \emph{tree-width} of $G$,
denoted $tw(G)$, is the smallest tree-width of any tree
decomposition of $G$.

We use the following three well-known  and easy results regarding
tree-width.

\begin{observation}\label{obs:tw-chi}
For a graph $G$, $\chi(G) \leq tw(G) +1$.
\end{observation}
\begin{proof}
Let $(T,t)$ be a tree decomposition of $G$ with $tw(T,t) = tw(G)$.
We create a supergraph $G'$ by completing the bags of $(T,t)$ to
cliques. This resulting graph is chordal, and thus perfect. So we now
have $\chi(G) \leq \chi(G') = \omega(G') = tw(G) + 1$.
\end{proof}

\begin{lemma}\label{lem:tw-cliquedecomp}
If $G$ contains a clique cutset $S$ where $G_1, \ldots, G_k$ are
the components of $G - S$, then:
\[
tw(G) = \max_{i} tw(G[G_i \cup S])
\]
\end{lemma}

\begin{lemma}\label{lem:tw-join}
Let $G$ be a graph that is the join of a graph $B$ and a clique $K$. Then
\[
tw(G) = tw(B) + |K|
\]
where $|K|$ denotes the number of vertices of $K$.
\end{lemma}

For further information on tree-width, see \cite{Reed}.

\begin{theorem}\label{thm:tw}
A (pan, even hole)-free graph $G$ has $tw(G)+1 \leq
1.5 \, \omega(G)$.
\end{theorem}
\begin{proof}
By Lemmas~\ref{lem:tw-cliquedecomp} and~\ref{lem:tw-join} and
Theorem~\ref{thm:structure_buoy}, we only need  show that $tw(B)
\leq 1.5 \, \omega(B)$ for any buoy in $G$. Recall that, for every bag
$B_i$ of $B$, either $B_i \cup B_{i+1}$ is a clique or $B_{i-1} \cup B_i$
is a clique. In particular, we can build a tree representation
$(T,t)$ of $B$ where $T$ is path using the unit circular-arc
construction from the proof of Theorem \ref{thm:buoy-unit-circ}.
To do this we choose the smallest $B_i$, and ``split'' the unit
circular-arc representation at the point $(i)$ and ``unroll'' it
onto a line. We now have a path where every point from our unit
circular-arc representation is a bag, and the extreme bags are
copies of the bag corresponding to the point $(i)$. Thus, by
adding the vertices of $B_i$ to every bag on this path, we obtain
a tree representation of $B$. It is easy to see that the largest
bag in this representation has size $\omega(G)+|B_i| \leq 1.5
\omega(G)$.
\end{proof}

Theorem \ref{thm:tw} is tight since odd cycles  have tree-width two
and clique number two. Similarly, by making a buoy with an odd
number of bags such that each bag has $k$ vertices and $B_i \cup
B_{i+1}$ is a clique for every $i$, we have a graph whose
tree-width is $3k-1$ and whose clique number is $2k$. (See
Figure~\ref{fig:non-beta-perfect} for an example with $k=2$.)
Moreover, by Observation \ref{obs:tw-chi}, we obtain the following
corollary.

\begin{corollary}\label{cor:chi-boundedness}
A (pan, even hole)-free graph $G$ has $\chi(G) \leq 1.5\omega(G)$.
\hfill $\Box$
\end{corollary}

\section{Conclusion and open problems}\label{sec:conclusions}
In this paper, we studied the structure of (claw, even hole)-free
graphs. It turned out that our results apply to the larger class of
(pan, even hole)-free graphs. From the structure results, we
obtained fast recognition and coloring algorithms for (pan, even
hole)-free graphs. The complexity of coloring even-hole-free
graphs is unknown.  It follows from Corollary~1 in
\cite{KraKra2001} that coloring odd-hole-free graphs is
NP-Complete. Thus, the following problem, analogous to our result,
is of interest to us.
\begin{problem}\label{pro:odd-hole}
What is the complexity of coloring (pan, odd hole)-free graphs?
\end{problem}
Observation~\ref{obs:tw-chi} shows the tree-width of a (pan, even
hole)-free graphs is bounded by a function in the clique number.
It is conceivable that a more general statement holds.
\begin{problem}\label{pro:tree-width}
Is the tree-width of an even-hole-free graph bounded by a function
of its clique number?
\end{problem}

The {\it clique-width} of a graph $G$, denoted by $cw(G)$, is the
minimum number of labels needed to construct $G$ using the
following four operations:
\begin{description}
 \item[(i)] Creation of a new vertex $v$ with label $i$.
 \item[(ii)] Disjoint union of two labeled graphs.
 \item[(iii)] Joining each vertex with
label $i$ to each vertex with label $j$.
 \item[(iv)] Changing label $i$ to $j$.

\end{description}
It is known \cite{CorRot2005} that for any graph $G$, $cw(G) \leq
3 \cdot 2^{tw(G)-1}$ and that \cite{CouOla2000} $cw(\overline{G})
\leq 2 \cdot cw(G)$ where $\overline{G}$ is the complement of $G$.

In \cite{CouMak2000}, it is shown that every problem definable in
a certain kind of Monadic Second Order Logic, called
LinEMSOL($\tau_1, L$) is linear-time
solvable on any graph class with bounded clique-width for which a
$k$-expression can be constructed in linear time. In
\cite{CouMak2000}, it is mentioned that, roughly speaking,
MSOL($\tau_1$) is Monadic Second Order Logic with quantification
over subsets of vertices but not of edges; MSOL($\tau_1, L$) is
the restriction of MSOL($\tau_1$) with the addition of labels
added to the vertices, and LinEMSOL($\tau_1, L$) is the
restriction of MSOL($\tau_1, L$) which allows search for sets
of vertices which are optimal with respect to some linear
evaluation functions. The problems Vertex Cover, Maximum Weight
Stable Set, Maximum Weight Clique, Steiner Tree and Domination are
examples of LinEMSOL($\tau_1, L$) definable problems. Furthermore,
from the results of \cite{Oum2008} and \cite{KobRot2003}, it
follows that the chromatic number of any class of graphs with
bounded clique-width can be computed in polynomial time.

In \cite{MakRot1999}, it is shown that split graphs have unbounded
clique-width. It follows that even-hole-free graphs have unbounded
clique-width. However, it might be possible that the clique-width
of an even-hole-free graph is bounded by a function of its clique
number. To conclude our paper, we pose this as an open problem.
\begin{problem}\label{pro:clique-width}
Is the clique-width of an even-hole-free graph bounded by a
function of its clique number?
\end{problem}

\bibliographystyle{plain}

\end{document}